\documentclass[final,1p,times]{elsarticle}

\usepackage{amsmath,amssymb,amsfonts,amscd,amsthm}
\usepackage{graphicx}

\usepackage{color,soul}

\def\R{\mathbb R}
\def\e{\mathrm e}
\def\O{\mathcal O}
\renewcommand{\i}{{\mathrm i}}
\def\d{\mathrm d}
\def\S{\mathcal S}
\def\diag{\mathrm{diag}}

\newtheorem{lem}{Lemma}[section]
\newtheorem{prop}[lem]{Proposition}
\newtheorem{theorem}[lem]{Theorem}

\newtheorem{de}[lem]{Definition}
\newtheorem{rem}[lem]{Remark}

\begin{document}

\begin{frontmatter}

\title{Fulop-Tsutsui interactions on quantum graphs}
\author[label1]
{Taksu Cheon \corref{cor1}}
\ead{taksu.cheon@kochi-tech.ac.jp}
\author[label1]
{Ond\v rej Turek}
\ead{ondrej.turek@kochi-tech.ac.jp}
\address
[label1]
{Laboratory of Physics, Kochi University of Technology,
Tosa Yamada, Kochi 782-8502, Japan}
\cortext[cor1]{corresponding author}

\date{July 1, 2010}

\begin{abstract}
We examine  scale invariant Fulop-Tsutsui couplings in a quantum vertex of a general degree $n$.
We demonstrate that essentially same scattering amplitudes as for the free coupling
can be achieved for two $(n-1)$-parameter Fulop-Tsutsui subfamilies 
if $n$ is odd, and for three $(n-1)$-parameter Fulop-Tsutsui subfamilies if $n$ is even.
We also work up an approximation scheme for a general Fulop-Tsutsui vertex, 
using only $n$ $\delta$ function potentials.
\end{abstract}


\end{frontmatter}


\section{Introduction}
%

The scale invariant interactions on graph, originally 
discovered by Fulop and Tsutsui 
in $n=2$ singular vertex \cite{FT00}, 
occupies a special position in the theory of quantum graphs \cite{EKST08}.
Unlike conventional $\delta$-like interaction, Fulop-Tsutsui interaction retains 
the invariance of the classical counterpart.
One of its interesting features is that it has no classical limit in usual sense, since
there is no parameter in the system bearing the scale of Planck constant. 
This fact makes a system with Fulop-Tsutsui interaction 
a potential quantum device for partial
transmission in a macroscopic scale.
%
%
Unlike the simplest case of $n=2$, 
two distinct Fulop-Tsutsui subfamilies are found
in the case of $n=3$ singular vertex \cite{CET09}.
One intriguing corollary is that there are two distinct ``free-like'' connection
condition for $n=3$ quantum graph.
A classification and parameterization of Fulop-Tsutsui couplings up to the degree $4$ can be 
found in a recent paper~\cite{Be1}.
%

In this letter, we examine the Fulop-Tsutsui subfamilies of singular
vertex of general degree $n$.
For a given $n$, we find $(n-1)$ different types of Fulop-Tsutsui subfamilies, 
which are classified by the ranks of the boundary matrices.
Through the examination of scattering matrices, we identify those subfamilies 
that contain the free-like connection conditions \cite{ES89}.  Two distinct free-like conditions 
are found for odd $n$ and three for even $n$.
We also reexamine the finite approximation of singular vertices, and
find that, for Fulop-Tsutsui subfamily of degree $n$, a construction 
with $n$ rescaled $\delta$s is possible, which is considerably simpler than 
the construction of general vertex requiring $n^2$ parameters \cite{CET10}.

\section{Formalism}
The singular vertex in a star graph of degree $n$ is described by
the connection condition of the form  \cite{KS99}
\begin{equation}\label{AB}
A\Psi(0) + B\Psi'(0) = 0
\end{equation}
where
\begin{eqnarray}
\Psi(x) = 
  \begin{pmatrix} \psi_1(x) \cr \vdots \cr \psi_n(x) \end{pmatrix} 
\end{eqnarray}
is an $n$ component vector with each entry $\psi_j(x)$ representing 
the wavefunction on the  $j$-th line coming out from the vertex $x=0$,
and $A$, $B$ are $n \times n$ matrices
with the properties $(AB)^\dagger=B^\dagger A$ and ${\rm rank}(A,B)=n$.
Fulop-Tsutsui interaction is a subset of this generic connection condition
with a property that its scattering matrix is constant as a function of 
the momentum of incoming particle.
Since $\Psi(0)$ and $\Psi'(0)$ have different dimensional scales, we
need to have no coupling between them in the connection condition.
In this paper we will work with so-called $ST$-form of boundary conditions introduced in \cite{CET10} -- in this form, the Fulop-Tsutsui subfamily of singular vertex is characterized by the connection condition
\begin{eqnarray}
\label{STF}
\begin{pmatrix} I^{(m)} & T \cr 0 & 0 \end{pmatrix} \Psi'(0)
=
\begin{pmatrix} 0 & 0 \cr -T^\dagger & I^{(n-m)} \end{pmatrix} \Psi(0)
\end{eqnarray}
where the symbol $I^{(m)}$ represents the identity matrix $m\times m$; the value $m$ corresponds to $\mathrm{rank}(B)$ from \eqref{AB}.
First $m$ equations of this connection condition involve only $\Psi'(0)$
and the last $(n-m)$ lines only $\Psi(0)$, and the $m \times (n-m)$ 
complex matrix $T=\left( t_{ij}\right)$ contains only dimensionless parameters.
Since $m$ can run from $1$ to $n-1$, there are $n-1$ distinct Tsutsui-Fulop
subfamilies, we index them by the ranks of $A$ 
and $B$:
\begin{eqnarray}
( {\rm rank}(A), {\rm rank}(B) ) = (n-1, 1), (n-2, 2), ... .., (1, n-1) .
\end{eqnarray}

Let us consider a quantum particle coming in from the $j$-th line scattered at
the singular vertex into all lines.  The process is described by the wavefunction
vector as 
\begin{eqnarray}
\Psi^{(j)}(x) = 
  \begin{pmatrix} 0 \cr \vdots \cr 1_j \cr \vdots \cr 0 \end{pmatrix} 
  e^{-{\textrm i} kx}
+ \begin{pmatrix} {\cal T}_{1j} \cr \vdots \cr {\cal R}_{j} \cr \vdots \cr {\cal T}_{nj} \end{pmatrix} 
  e^{{\textrm i} kx}
\end{eqnarray}
where ${\cal R}_j$ is the reflection amplitude for the line $j$, and ${\cal T}_{ij}$, 
is the transmission amplitude from the $j$-th to the $i$-th line.
We can collect $\Psi^{(j)}(0)$ and $\Psi'^{(j)}(0)$ with $j=1, ..., n$ to form 
$n \times n$ matrices
\begin{eqnarray}
&&\!\!\!\!\!\!\!\!
\left(\Psi^{(1)}(0) \cdots \Psi^{(n)}(0)\right) = {\cal S}(k) + I,
\nonumber \\
&&\!\!\!\!\!\!\!\!
\left(\Psi'^{(1)}(0) \cdots \Psi'^{(n)}(0)\right) = {\textrm i}k ( {\cal S}(k) - I ) .
\end{eqnarray}
Here, ${\cal S}$ is the scattering matrix given by
\begin{eqnarray}
\label{ScatMat}
{\cal S}(k)=\begin{pmatrix} 
 {\cal R}_{1}(k)   & {\cal T}_{12}(k) & \cdots & {\cal T}_{1n}(k) \cr 
 {\cal T}_{21}(k) & {\cal R}_{2}(k)   & \cdots & {\cal T}_{2n}(k) \cr
 \vdots & & & \vdots \cr
 {\cal T}_{n1}(k) & {\cal T}_{n2}(k)& \cdots & {\cal R}_{n}(k) \end{pmatrix} .
\end{eqnarray}
We can combine the connection conditions (\ref{STF}) for 
all $\Psi^{(j)}(x)$, $j=1,\ldots, n$, 
and after trivial rearrangement, to get rid of the spurious $\textrm{i}k$ factor, 
we obtain
%
\begin{eqnarray}
\label{ST}
\begin{pmatrix} I^{(m)} & T \cr  T^\dagger & -I^{(n-m)}  \end{pmatrix} 
{\cal S}
=
\begin{pmatrix} I^{(m)} & T \cr -T^\dagger & I^{(n-m)} \end{pmatrix}\,. 
\end{eqnarray}
The scattering matrix can be expressed in the explicit form
%
\begin{equation}\label{S}
\S=\left(\begin{array}{cc}
\left(I^{(m)}+TT^\dagger\right)^{-1} \left(I^{(m)}-TT^\dagger\right) & \left(I^{(m)}+TT^\dagger\right)^{-1}2 T \\
\left(I^{(n-m)}+T^\dagger T\right)^{-1}  2 T^\dagger & -\left(I^{(n-m)}+T^\dagger T\right)^{-1} \left(I^{(n-m)}-T^\dagger T\right)
\end{array}\right)\,.
\end{equation}
%
As anticipated, the scattering matrix has no $k$-dependence,
and a quantum wave on a line, of any wave number, is branched to other lines
with fixed ratio ${\cal S}_{ij}$.
For the Fulop-Tsutsui subfamily of $(n-m, m)$,  $T$ is a
complex matrix of dimensions $m\times(n-m)$, thus offers $2m(n-m)$ parameters.
This number becomes maximum at $m=[\frac{n}{2}]$ to give $[\frac{n^2}{2}]$.
Note that the scattering matrix ${\cal S}$ has $2n + 2 \frac{n^2-n}{2}$ 
$= n^2$ parameters.  Therefore, the full tuning of ${\cal S}$ 
with Fulop-Tsutsui interaction is not attainable.

\begin{rem}\label{shift}
Throughout the whole paper, the elements of the matrix $T$ arising in the $ST$-form~\eqref{ST} will be denoted $t_{jl}$ for $j=1,\ldots,m$ and $l=m+1,\ldots,n$ (note well the shift of the column indices: we use numbers $m+1,\ldots,n$ instead of $1,\ldots,n-m$).
\end{rem}

%
\section{Free and free-like connections}

A prominent example of the Fulop-Tsutsui interaction is the \emph{free connection} which is given by the boundary conditions
\begin{eqnarray}
\psi_1(0)=\psi_2(0)=\cdots=\psi_n(0)=:\psi(0) \, , 
\qquad \sum_{j=1}^{n}\psi_j'(0)=0\,.
\end{eqnarray}
It is easy to see that its $ST$-form \eqref{STF} is obtained by setting $m=1$ and $T=\left(\begin{array}{cccc}1&1&\cdots&1\end{array}\right)$. Therefore, with regard to \eqref{ScatMat}, the scattering properties of the free connection are described by the scattering matrix
\begin{eqnarray}
\label{frees}
{\cal S} = 
\begin{pmatrix} 
  \frac{2}{n}-1   & \frac{2}{n} & \cdots & \frac{2}{n} \cr 
  \frac{2}{n} & \frac{2}{n}-1   &  & \vdots \cr
  \vdots &    & \ddots  &  \frac{2}{n} \cr
  \frac{2}{n} & \cdots & \frac{2}{n} & \frac{2}{n}-1 
\end{pmatrix} .
\end{eqnarray}
Since it will be useful many times in the rest of the paper, we introduce here the symbol $J^{(p,q)}$ that represents the rectangular matrix $p\times q$ all of whose elements are equal to $1$. If $p=q$, we will abbreviate the notation to $J^{(p)}$. Using this symbol, \eqref{frees} can be rewritten as $\S=\frac{2}{n}J^{(n)}-I^{(n)}$.

We see immediately from \eqref{frees} that the absolute values of 
the scattering amplitudes are equal to
\begin{equation}\label{RT}
|{\cal R}_j|=1-\frac{2}{n}\qquad\text{and}\qquad|{\cal T}_{ij}|=\frac{2}{n}
\end{equation}
for all $i,j=1,\ldots,n$.

Now we proceed to the free-like connection. We start our considerations at the definition.

\begin{de}
Any vertex coupling such that its scattering amplitudes satisfy \eqref{RT} is called \emph{free-like vertex coupling}, and the corresponding boundary conditions are called \emph{free-like boundary conditions}.
\end{de}

Apart from the phases of the reflection and transmission coefficients ${\cal T}_{ij}$ and ${\cal R}_j$,
which themselves are not observable,
all the free-like boundary conditions give the physical amplitude identical to the free case.
In the rest of this section, we aim to find all free-like vertex couplings. We will derive a full and at the same time simple characterization of this subfamily in terms of boundary conditions. 

In the first proposition, we answer the question what structure the \emph{scattering matrix} of a free-like vertex coupling can have.

\begin{prop}\label{possibleS}
Let the scattering matrix $\mathcal{S}$ in a vertex of degree $n$ satisfies the conditions~\eqref{RT}, i.e. the diagonal and off-diagonal terms of $\mathcal{S}$ are of moduli $1-\frac{2}{n}$ and $\frac{2}{n}$, respectively. Then $\mathcal{S}$ falls within one of the following three cases:
\begin{itemize}
\item $\S=D^{-1}\left(-I^{(n)}+\frac{2}{n}J^{(n)}\right)D$ for $D=\diag\left(1,\e^{\i\xi_2},\ldots,\e^{\i\xi_n}\right)$, where $\xi_j\in\R$ for all $j=2,\ldots,n$,
\item $\S=D^{-1}\left(I^{(n)}-\frac{2}{n}J^{(n)}\right)D$ for $D=\diag\left(1,\e^{\i\xi_2},\ldots,\e^{\i\xi_n}\right)$, where $\xi_j\in\R$ for all $j=2,\ldots,n$,
\item $n$ is even and $\S=P^{-1}D^{-1}\left(\begin{array}{c|c}I^{(\frac{n}{2})}-\frac{2}{n}J^{(\frac{n}{2})}&\frac{2}{n}J^{(\frac{n}{2})}\\\hline\frac{2}{n}J^{(\frac{n}{2})}&-I^{(\frac{n}{2})}+\frac{2}{n}J^{(\frac{n}{2})}\end{array}\right)DP$, where
$P$ is a permutation matrix and $D=\diag\left(1,\e^{\i\xi_2},\ldots,\e^{\i\xi_n}\right)$, where
$\xi_j\in\R$ for all $j=2,\ldots,n$.
\end{itemize}
\end{prop}

\begin{proof}
Clearly, any connection condition resulting in free-like scattering amplitude has to belong
to the Fulop-Tsutsui family, because any other connection condition results 
in $k$-dependence \cite{FKW07}.
Consequently, the scattering matrix $\S$ has the form \eqref{S}. It is important to note that any scattering matrix $\S$ is unitary, cf.~\cite{KS00,Ha00}, and so is $\S$. In the case of the Fulop-Tsutsui coupling, $\S$ is also Hermitean, which is obvious from~\eqref{S}. Therefore we can write, with regard to the scattering amplitudes \eqref{RT},
%
%
\begin{eqnarray}
\mathcal{R}_j=\pm\left(1-\frac{2}{n}\right)\qquad\text{and}\qquad \mathcal{T}_{ij}=\frac{2}{n}\,\e^{\i\alpha_{ij}}
\end{eqnarray}
for $\alpha_{ij}=-\alpha_{ji}$.

Let $p$ denote the number of diagonal elements of $\S$ equal to $1-\frac{2}{n}$, consequently $\S$ has $n-p$ diagonal elements equal to $\frac{2}{n}-1$. We divide the rest of the proof into three parts according to the value of $p$.

\emph{Case $p=0$.} Let us suppose that $p=0$, i.e. $\mathcal{R}_j=\frac{2}{n}-1$ for all $j=1,\ldots,n$. We define
\begin{eqnarray}
D=\diag\left(1,\e^{\i\alpha_{12}},\e^{\i\alpha_{13}},\ldots,\e^{\i\alpha_{1n}}\right)
\end{eqnarray}
and set $\hat{\S}=D\S D^{-1}$. Matrix $\hat{\S}$ has the following properties:
\begin{itemize}
\item all its diagonal elements are equal to $\frac{2}{n}-1$,
\item all off-diagonal elements in the first row and in the first column are equal to $\frac{2}{n}$,
\item for any $i,j\in\{2,\ldots,n\}$, the $(i,j)$-th element of $\hat{\S}$ equals $\frac{2}{n}\e^{\i\left(\alpha_{ij}+\alpha_{1i}-\alpha_{1j}\right)}$,
\item $\hat{\S}$ is unitary, since both $\S$ and $D$ are unitary.
\end{itemize}
Let $i\in\{2,\ldots,n\}$. The scalar product of the first row of $\hat{\S}$ and the $i$-th one equals
\begin{multline}
\sum_{j\in\{2,\ldots,n\},j\neq i} \frac{2}{n}\cdot\frac{2}{n}\,\e^{\i\left(\alpha_{ij}+\alpha_{1i}-\alpha_{1j}\right)}+\left(\frac{2}{n}-1\right)\cdot\frac{2}{n}+\frac{2}{n}\cdot\left(\frac{2}{n}-1\right)\\
=-\frac{4}{n^2}\left(n-2-\sum_{j\in\{2,\ldots,n\},j\neq i}\e^{\i\left(\alpha_{ij}+\alpha_{1i}-\alpha_{1j}\right)}\right)\,.
\end{multline}
This number has to equal $0$ due to the unitarity of $\hat{\S}$, hence we deduce that $\e^{\i\left(\alpha_{ij}+\alpha_{1i}-\alpha_{1j}\right)}=1$ for all $i,j=2,\ldots,n$, $i\neq j$. In other words, \emph{all off-diagonal elements of $\hat{\S}$ are equal to $\frac{2}{n}$}, i.e. $\hat{\S}=\frac{2}{n}J^{(n)}-I^{(n)}$. Consequently, $S=D^{-1}\hat{\S}D=D^{-1}\left(\frac{2}{n}J^{(n)}-I^{(n)}\right)D$.

\emph{Case $p=n$.} If $p=n$, it suffices to choose
\begin{eqnarray}
D=\diag\left(1,\e^{\i(\alpha_{12}+\pi)},\e^{\i(\alpha_{13}+\pi)},\ldots,\e^{\i(\alpha_{1n}+\pi)}\right)\,,
\end{eqnarray}
and proceed similarly like in the previous case, $p=0$.

\emph{Case $1\leq p\leq n-1$.} Now matrix $\S$ contains at least one diagonal element equal to $1-\frac{2}{n}$ and at least one diagonal element equal to $\frac{2}{n}-1$. In this case, we begin with a rearrangement of $\S$, done by a simultaneous permutation of the rows and columns of $\S$. Let $P$ be such a permutation matrix that
\begin{eqnarray}
\left(P^{-1}\S P\right)_{jj}=1-\frac{2}{n} \quad \text{for $j=1,\ldots,p$}\,;
\end{eqnarray}
in what follows we use the notation $\S_P=P\S P^{-1}$.

In the next step, we define
\begin{eqnarray}
\hat{D}=\diag\left(1,\e^{\i(\alpha_{12}+\pi)},\ldots,\e^{\i(\alpha_{1p}+\pi)},\e^{\i\alpha_{1,p+1}},\ldots,\e^{\i\alpha_{1n}}\right)
\end{eqnarray}
and set $\hat{\S}=\hat{D}\S_P\hat{D}^{-1}$. The following properties of $\hat{\S}$ are important for our considerations:
\begin{itemize}
\item $(\hat{S})_{jj}=1-\frac{2}{n}$ for $j=1,\ldots,p$; $(\hat{S})_{jj}=\frac{2}{n}-1$ for $j=p+1,\ldots,n$,
\item the first row of $\hat{S}$ and the transposed first column are equal to
\begin{eqnarray}
\left(
\underbrace{\left(1-\frac{2}{n}\right)\qquad-\frac{2}{n}\qquad\ldots\qquad-\frac{2}{n}}_{p}\qquad\underbrace{\frac{2}{n}\qquad\ldots\qquad\frac{2}{n}}_{n-p}
\right)\,,
\end{eqnarray}
\item for any $i,j\in\{2,\ldots,n\}$, $i\neq j$, the $(i,j)$-th element of $\hat{\S}$ has modulus $\frac{2}{n}$, 
\item the $(p+1)$-st row of $\hat{S}$ is equal to
\begin{multline}
\!\!\!\!
\left(\underbrace{-\frac{2}{n}\e^{\i\left(\alpha_{p+1,l}+\alpha_{1,p+1}-\alpha_{1l}\right)}\quad\ldots\quad-\frac{2}{n}\e^{\i\left(\alpha_{p+1,p}+\alpha_{1,p+1}-\alpha_{1p}\right)}}_{p}\right.\\
\left.\underbrace{\left(\frac{2}{n}-1\right)\quad\frac{2}{n}\e^{\i\left(\alpha_{p+1,p+2}+\alpha_{1,p+1}-\alpha_{1,p+2}\right)}\quad\ldots\quad\frac{2}{n}\e^{\i\left(\alpha_{p+1,n}+\alpha_{1,p+1}-\alpha_{1n}\right)}}_{n-p}\right)\,,
\end{multline}
\item $\hat{\S}$ is unitary, since $\hat{\S}=\hat{D}P\S P^{-1}\hat{D}^{-1}$ and $\S$, $P$, $\hat{D}$ are unitary,
\item $\hat{\S}$ is Hermitean, since $\hat{\S}=\hat{D}P\S P^{-1}\hat{D}^{-1}=\hat{D}P\S P^{\dagger}\hat{D}^{\dagger}$ and $\S$ is Hermitean.
\end{itemize}
Let $j\in\{2,\ldots,p\}$. Similarly as in the case $p=n$, we deduce from the scalar product of the first row of $\hat{\S}$ and the $j$-th one ($j=2,\ldots,p$) that the first $p$ rows of $\hat{\S}$ can be written in this way:
\begin{eqnarray}
\left(\begin{array}{c|c}
I^{(p)}-\frac{2}{n} J^{(p)} & \frac{2}{n} J^{(p,n-p)}
\end{array}\right)\,.
\end{eqnarray}

Since $\hat{S}$ is Hermitean, lower left submatrix of $\hat{S}$ is equal to the transposed upper right submatrix, therefore
\begin{eqnarray}
\hat{\S}=\left(\begin{array}{c|cccc}
I^{(p)}-\frac{2}{n}J^{(p)} &  & \frac{2}{n}J^{(p,n-p)} &  &  \\
\hline
 & \frac{2}{n}-1 & \frac{2}{n}\e^{\i\left(\alpha_{p+1,p+2}+\alpha_{1,p+1}-\alpha_{1,p+2}\right)} & \cdots & \frac{2}{n}\e^{\i\left(\alpha_{p+1,n}+\alpha_{1,p+1}-\alpha_{1n}\right)} \\
\frac{2}{n}J^{(n-p,p)} & \frac{2}{n}\e^{\i\left(\alpha_{p+2,p+1}+\alpha_{1,p+2}-\alpha_{1,p+1}\right)} & 1-\frac{2}{n} & \cdots & \frac{2}{n}\e^{\i\left(\alpha_{p+2,n}+\alpha_{1,p+2}-\alpha_{1n}\right)} \\
 & \vdots & \vdots & \ddots & \vdots \\
 & \frac{2}{n}\e^{\i\left(\alpha_{p+2,p+1}+\alpha_{1,p+2}-\alpha_{1,p+1}\right)} & \frac{2}{n}\e^{\i\left(\alpha_{n,p+2}+\alpha_{1,p+2}-\alpha_{1n}\right)} & \cdots & 1-\frac{2}{n}
\end{array}\right)\,.
\end{eqnarray}

Now we introduce another diagonal matrix,
\begin{eqnarray}
\check{D}=\diag\left(\underbrace{1,\ldots,1}_{p},1,\e^{\i\left(\alpha_{p+1,p+2}+\alpha_{1,p+1}-\alpha_{1,p+2}\right)},\e^{\i\left(\alpha_{p+1,p+3}+\alpha_{1,p+1}-\alpha_{1,p+3}\right)},\ldots,\e^{\i\left(\alpha_{p+1,n}+\alpha_{1,p+1}-\alpha_{1n}\right)}\right)\,,
\end{eqnarray}
and define $\check{\S}=\check{D}\hat{\S}\check{D}^{-1}$; it is easy to see that the matrix $\check{S}$ has the following properties:
\begin{itemize}
\item the diagonal elements satisfy
\begin{eqnarray}
\left(\check{\S}\right)_{jj}=1-\frac{2}{n}\quad\text{for $j=1,\ldots,p$}\,,\qquad \left(\check{\S}\right)_{jj}=\frac{2}{n}-1\quad\text{for $j=p+1,\ldots,n$}\,,
\end{eqnarray}
\item the off-diagonal elements in the $(p+1)$-st row and in the $(p+1)$-st column are equal to $\frac{2}{n}$,
\item all other off-diagonal elements are of modulus $\frac{2}{n}$, and therefore we may (and will) write $\left(\check{\S}\right)_{jl}=\frac{2}{n}\,\e^{\i\beta_{jl}}$ for certain $\beta_{jl}$,
\item $\check{\S}$ is unitary, since both $\hat{S}$ and $\check{D}$ are unitary,
\item $\check{\S}$ is Hermitean, since $\check{\S}=\check{D}\hat{\S}\check{D}^{-1}=\check{D}\hat{\S}\check{D}^{\dagger}$ and $\hat{S}$ is Hermitean.
\end{itemize}
For any $i>p+1$, the scalar product of the $(p+1)$-st row and the $i$-th row equals
\begin{eqnarray}
-\frac{4}{n^2}\left(n-2-\sum_{j\in\{1,\ldots,n\}\backslash\{p+1,i\}}\e^{\i\beta_{ij}}\right)\,,
\end{eqnarray}
hence, with regard to the unitarity, $\e^{\i\beta_{ij}}=1$ for all $i>p+1$ and $j\neq p+1,j\neq i$. Moreover, since $\check{\S}$ is Hermitean, it holds also $\e^{\i\beta_{ji}}=1$ for all $i>p+1$, $j\neq p+1,j\neq i$. To sum up, we have
\begin{eqnarray}
\check{S}=\left(\begin{array}{c|c}
I^{(p)}-\frac{2}{n}J^{(p)} & \frac{2}{n}J^{(p,n-p)} \\
\hline
\frac{2}{n}J^{(n-p,p)} & -I^{(n-p)}+\frac{2}{n}J^{(n-p)}
\end{array}\right)\,.
\end{eqnarray}
Finally, the scalar product of the first and the $(p+1)$-st line is equal to $(n-2p)\frac{4}{n^2}$, and since it has to vanish due to the unitarity of $\check{S}$, it holds $p=\frac{n}{2}$, which also requires that $n$ is even.
Now we denote $\check{D}\hat{D}$ as $D$, and since it holds $D_{11}=1$, the statement is proved.
\end{proof}

\begin{rem}
It is easy to realize that the position of the number $1$ in the matrices $D$ in the statements (a), (b) and (c) of Proposition~\ref{possibleS} is not important. We have put $1$ at the position $(1,1)$ of $D$, but it can stand wherever on the diagonal.
\end{rem}

Proposition~\ref{possibleS} describes the set of free-like vertex couplings in terms of scattering matrices. Now we will find the corresponding \emph{boundary conditions}, i.e. we show how to set the matrix $T$ in \eqref{ST} to achieve the scattering matrices listed in Proposition~\ref{possibleS}. Theorem~\ref{FreeLike} will fully characterize the subfamily of free-like boundary conditions.

\begin{theorem}\label{FreeLike}
Let the boundary conditions in a vertex of degree $n$ be described by the $ST$-form~\eqref{ST}. It holds:
\begin{itemize}
\item[(a)] If $m=1$ and $T=\left(\begin{array}{ccc}\e^{\i\xi_2}&\cdots&\e^{\i\xi_n}\end{array}\right)$ where $\xi_j\in\R$ for all $j=2,\ldots,n$, then
\begin{eqnarray}
\S=D^{\dagger}\left(-I^{(n)}+\frac{2}{n}J^{(n)}\right)D \qquad \text{for}\quad D=\diag\left(1,\e^{\i\xi_2},\ldots,\e^{\i\xi_n}\right)\,.
\end{eqnarray}
\item[(b)] If $m=n-1$ and $T=\left(\begin{array}{c}\e^{\i\xi_2}\\ \vdots\\\e^{\i\xi_n}\end{array}\right)$ where $\xi_j\in\R$ for all $j=2,\ldots,n$, then
\begin{eqnarray}
\S=D^{\dagger}\left(I^{(n)}-\frac{2}{n}J^{(n)}\right)D \qquad \text{for}\quad D=\diag\left(1,\e^{\i\xi_2},\ldots,\e^{\i\xi_n}\right)\,.
\end{eqnarray}
\item[(c)] Let $n$ be even, $m=\frac{n}{2}$ and $T=\frac{2}{n}X^{\dagger}J^{(\frac{n}{2})}Y$ where $X=\diag\left(1,\e^{\i\xi_2},\ldots,\e^{\i\xi_\frac{n}{2}}\right)$,\\
$Y=\diag\left(\e^{\i\xi_{\frac{n}{2}+1}},\ldots,\e^{\i\xi_n}\right)$ and $\xi_j\in\R$ for all $j=2,\ldots,n$. Let the indexes of the outgoing edges be reindexed by a permutation represented by a permutation matrix $P$. Then the scattering matrix corresponding to the system with reindexed edges equals
\begin{eqnarray}
\S=P^{-1}D^{\dagger}\left(\begin{array}{c|c}-I^{(\frac{n}{2})}+\frac{2}{n}J^{(\frac{n}{2})}&-\frac{2}{n}J^{(\frac{n}{2})}\\\hline-\frac{2}{n}J^{(\frac{n}{2})}&I^{(\frac{n}{2})}-\frac{2}{n}J^{(\frac{n}{2})}\end{array}\right)DP \qquad \text{for}\quad D=\diag\left(1,\e^{\i\xi_2},\ldots,\e^{\i\xi_n}\right)\,.
\end{eqnarray}
%
\end{itemize}
In all the cases listed above, the reflexion and transmission amplitudes satisfy
\begin{equation}\label{RT.}
|{\cal R}_j|=1-\frac{2}{n} \qquad\text{and}\qquad |{\cal T}_{ij}|=\frac{2}{n}\,.
\end{equation}
Moreover, the statement can be inverted: If the reflexion and transmission amplitudes satisfy \eqref{RT.}, then the system corresponds to one of the situations (a), (b), (c) described above.
\end{theorem}

\begin{proof}

\medskip
\noindent (a)\quad If $m=1$ and $T=\left(\begin{array}{ccc}\e^{\i\xi_2}&\cdots&\e^{\i\xi_n}\end{array}\right)$, then
\begin{equation}\label{TT}
T T^\dagger = \left(n-1\right)\qquad\text{and}\qquad
\left(T^\dagger T\right)_{ij} = \e^{\i(\xi_j-\xi_i)} \quad\text{for $i,j=2,\ldots,n$}\,.
\end{equation}
This can be used to calculate the inverses,
\begin{eqnarray}
&&\!\!\!\!\!\!\!\!\!\!\!\!
\left( I^{(1)} + T T^\dagger \right)^{-1} 
= \left(\frac{1}{n}\right)\,,
\nonumber \\
&&\!\!\!\!\!\!\!\!\!\!\!\!
\left( I^{(n-1)} + T^\dagger T \right)^{-1} 
= I^{(n-1)}-\frac{1}{n} T^\dagger T .
\end{eqnarray}
Four submatrices of ${\cal S}$ are given by
\begin{eqnarray}
&&\!\!\!\!\!\!\!\!\!\!\!\!
\left( I^{(1)} + T T^\dagger \right)^{-1} 2T
= \frac{2}{n}T\,,
\\ \nonumber
&&\!\!\!\!\!\!\!\!\!\!\!\!
\left( I^{(n-1)} + T^\dagger T \right)^{-1} 2T^\dagger
= \frac{2}{n} T^\dagger\,,
\end{eqnarray}
and 
\begin{eqnarray}
&&\!\!\!\!\!\!\!\!\!\!\!\!\!\!\!\!\!\!\!\!\!\!\!\!
\left( I^{(1)} + T T^\dagger \right)^{-1} \left( I^{(m)} - T T^\dagger \right)
= \left(-1+\frac{2}{n}\right)\,,
\\ \nonumber
&&\!\!\!\!\!\!\!\!\!\!\!\!\!\!\!\!\!\!\!\!\!\!\!\!\!\!\!\!
-\left( I^{(n-1)} + T^\dagger T \right)^{-1} \left( I^{(n-m)} + T^\dagger T \right)
= -I^{(n-1)}+\frac{2}{n} T^\dagger T\,.
\end{eqnarray}
If we employ the matrix $D$ defined in the part (a) of the Theorem, we obtain
\begin{eqnarray}
\S=\left(\begin{array}{cc}
-1+\frac{2}{n} & \frac{2}{n}T \\
\frac{2}{n}T^{\dagger} & -I^{(n-1)}+\frac{2}{n} T^\dagger T
\end{array}\right)=
D^{\dagger}\left(-I^{(n)}+\frac{2}{n}J^{(n)}\right)D\,.
\end{eqnarray}
Finally, Equalities~\eqref{RT.} hold due to Eq.~\eqref{TT}.

\medskip
\noindent (b)\quad If $m=n-1$ and $T=\left(\begin{array}{c}\e^{\i\xi_2}\\ \vdots\\\e^{\i\xi_n}\end{array}\right)$,
one can proceed in a way similar to the previous case. We have
\begin{equation}\label{TTn-1}
\left(T T^\dagger\right)_{ij} = e^{\i(\xi_i-\xi_j)} \quad\text{for $i,j=2,\ldots,n$}\,,
\quad
T^\dagger T = \left(n-1\right)
\end{equation}
and
\begin{equation*}
\left( I^{(n-1)} + T T^\dagger \right)^{-1} 
= I^{(n-1)}-\frac{1}{n} T T^\dagger\,,\qquad
\left( I^{(1)} + T^\dagger T \right)^{-1} 
= \left(\frac{1}{n}\right)\,,
\end{equation*}
then, with regard to the definition of the matrix $D$ from the part (b) of the Theorem, we have
\begin{eqnarray}
{\cal S}=
\left(\begin{array}{cc}
I^{(n-1)}-\frac{2}{n} T T^\dagger & \frac{2}{n} T \\
\frac{2}{n} T^\dagger & 1-\frac{2}{n}
\end{array}\right)=
D^{\dagger}\left(I^{(n)}-\frac{2}{n}J^{(n)}\right)D\,,
\end{eqnarray}
and \eqref{RT.} holds due to Eq.~\eqref{TTn-1}.

\medskip

\noindent (c)\quad Let us begin with the system with non-perturbed indexing of the edges. If $n$ is even, $m=\frac{n}{2}$ and $T=\frac{2}{n}X^{\dagger}J^{(\frac{n}{2})}Y$ for $X,Y$ defined above, it holds
\begin{eqnarray}
T T^\dagger = \frac{2}{n}X^{\dagger}J^{(\frac{n}{2})}X\,,
\qquad
T^\dagger T = \frac{2}{n}Y^{\dagger}J^{(\frac{n}{2})}Y\,.
\end{eqnarray}
It is easy to verify that
\begin{eqnarray}
&&\!\!\!\!\!\!\!\!\!\!\!\!
\left( I^{(\frac{n}{2})} + T T^\dagger \right)^{-1} 
= I^{(\frac{n}{2})}-\frac{1}{n} X^{\dagger}J^{(\frac{n}{2})}X\,,
\nonumber \\
&&\!\!\!\!\!\!\!\!\!\!\!\!
\left( I^{(\frac{n}{2})} + T^\dagger T \right)^{-1} 
= I^{(\frac{n}{2})}-\frac{1}{n} Y^{\dagger}J^{(\frac{n}{2})}Y\,.
\end{eqnarray}
Hence, the four submatrices of ${\cal S}$ are given by
\begin{eqnarray}
\label{offdiag}
&&\!\!\!\!\!\!\!\!\!\!\!\!
\left( I^{(\frac{n}{2})} + T T^\dagger \right)^{-1} 2T
= \frac{2}{n}X^{\dagger}J^{(\frac{n}{2})}Y\,,
\\ \nonumber
&&\!\!\!\!\!\!\!\!\!\!\!\!
\left( I^{(\frac{n}{2})} + T^\dagger T \right)^{-1} 2T^\dagger
= \frac{2}{n}Y^{\dagger}J^{(\frac{n}{2})}X
\end{eqnarray}
and 
\begin{eqnarray}
\label{diag}
&&\!\!\!\!\!\!\!\!\!\!\!\!\!\!\!\!\!\!\!\!\!\!\!\!
\left( I^{(\frac{n}{2})} + T T^\dagger \right)^{-1} \left( I^{(\frac{n}{2})} - T T^\dagger \right)
= I^{(\frac{n}{2})}-\frac{2}{n}X^{\dagger}J^{(\frac{n}{2})}X\,,
\\ \nonumber
&&\!\!\!\!\!\!\!\!\!\!\!\!\!\!\!\!\!\!\!\!\!\!\!\!
-\left( I^{(\frac{n}{2})} + T^\dagger T \right)^{-1} \left( I^{(\frac{n}{2})} - T^\dagger T \right)
= -I^{(\frac{n}{2})}+\frac{2}{n}Y^{\dagger}J^{(\frac{n}{2})}Y\,.
\end{eqnarray}
Therefore, the matrix $\S_0$ corresponding to the system with unperturbed indexing can be rearranged into the form
\begin{eqnarray}
{\cal S}_0=
\left(\begin{array}{cc}
I^{(\frac{n}{2})}-\frac{2}{n}X^{\dagger}J^{(\frac{n}{2})}X & \frac{2}{n}X^{\dagger}J^{(\frac{n}{2})}Y \\
\frac{2}{n}Y^{\dagger}J^{(\frac{n}{2})}X & -I^{(\frac{n}{2})}+\frac{2}{n}Y^{\dagger}J^{(\frac{n}{2})}Y
\end{array}\right)=
D^{\dagger}\left(\begin{array}{c|c}-I^{(\frac{n}{2})}+\frac{2}{n}J^{(\frac{n}{2})}&-\frac{2}{n}J^{(\frac{n}{2})}\\\hline-\frac{2}{n}J^{(\frac{n}{2})}&I^{(\frac{n}{2})}-\frac{2}{n}J^{(\frac{n}{2})}\end{array}\right)D\,,
\end{eqnarray}
where we have employed the matrix $D$ defined in the part (c) of the Theorem.

The perturbation of edge indexes result in a simultaneous perturbation of the columns and rows of the scattering matrix in the way $\S=P^{-1}\S_0P$, therefore
\begin{eqnarray}
{\cal S}=
P^{-1}D^{\dagger}\left(\begin{array}{c|c}-I^{(\frac{n}{2})}+\frac{2}{n}J^{(\frac{n}{2})}&-\frac{2}{n}J^{(\frac{n}{2})}\\\hline-\frac{2}{n}J^{(\frac{n}{2})}&I^{(\frac{n}{2})}-\frac{2}{n}J^{(\frac{n}{2})}\end{array}\right)DP\,.
\end{eqnarray}
It is easy to check that the diagonal elements of $\mathcal{S}$ are equal $1-\frac{2}{n}$ in modulus and the off-diagonal elements have moduli $\frac{2}{n}$, therefore \eqref{RT.} is satisfied.

\medskip

The last statement of Theorem~\ref{FreeLike} follows immediately from Proposition~\ref{possibleS}.
\end{proof}

\begin{rem}
If we restrict the system to be {\it time-reversal symmetric}, only real numbers are allowed 
for the elements of the matrix $T$, as well as for ${\cal R}$ and ${\cal T}$. This leaves the possible choice of phases to be either zero or $\pi$, therefore the $n-1$ parameters of the family of the free-like vertex coupling turn into exactly $2^{n-1}$ possibilities.
\end{rem}

\section{Finite approximation}
In this section we propose a way how a vertex with a general connecting condition 
of the Fulop-Tsutsui type can be approximated (and thus constructed) 
using only more regular objects, namely $\delta$-couplings 
and in some cases vector potentials. Our approximating model is 
based on the one presented in \cite{CET10}, but we modify it 
in order to make it fit better to the scale-invariant coupling. As a result, 
the number of the $\delta$-couplings that have to be involved is reduced to $n$; 
we remind that the model in \cite{CET10} needs generally $\frac{n(n+1)}{2}$ of them. 
The key idea allowing such a reduction,
which is a generalization of two-$\delta$ construction of 
Fulop-Tsutsui interaction in $n=2$ graph \cite{ZO10}, is explained below.

Let us consider a star graph with $n$ outgoing edges with the scale invariant vertex coupling given by the condition
\begin{subequations}\label{FTFT}
\begin{gather}
\psi_j'(0)+\sum_{l=m+1}^n t_{jl}\psi'_l(0)
=0 \qquad j=1, ..., m\,, \label{j<=m} \\
0=-\sum_{k=1}^m \overline{t_{kj}}\psi_k(0)+\psi_j(0)
\qquad j=m+1,\ldots,n \label{j>=m+1} \,.
\end{gather}
\end{subequations}
The columns of the matrix $T$ are indexed by the numbers $m+1,m+2,\ldots,n$ according to Remark~\ref{shift}. Thanks to this shift, the formulas will be simpler.

%
\noindent The approximating graph is constructed as follows.
\begin{itemize}
\item We take $n$ halflines, each parametrized by
$x\in[0,+\infty)$, with the endpoints denoted as $V_j$, $j=1,\ldots,n$.
\item For certain pairs $\{j,k\}$ ($j,k\in\{1,\ldots,n\}$), we join halfline
endpoints $V_j,V_k$ by a connecting edge of the length $\frac{d}{\gamma_{jk}}$, parametrized by $x\in[0,\frac{d}{\gamma_{jk}}]$. The criterion for which pairs $\{j,k\}$ the vertices $V_j,V_k$ are connected, as well as the values of the positive coefficients $\gamma_{jk}$, will be specified later. If the endpoints $V_j,V_k$ are not connected, we may set formally $\gamma_{jk}=0$, because it leads to the ``infinite length'' $\frac{d}{\gamma_{jk}}=\infty$. Note that $\gamma_{jk}=\gamma_{kj}$ for all $j,k=1,\ldots,n$.

The fact that the lenght of the connecting lines may vary represents the modification of the model \cite{CET10} that we have mentioned in the beginning of this section.
\item We put a $\delta$-coupling with the parameter $\alpha_j(d)$ at the vertex $V_j$ for each $j=1,..., n$.
\item On each connecting edge described above we put a constant vector
potential supported by the interval $\left[\frac{d}{4\gamma_{jk}},\frac{3d}{4\gamma_{jk}}\right]$. Its strength, which may be also zero, will be denoted by $A_{(j,k)}(d)$ (if the variable on the line grows from $V_j$ to $V_k$) or $A_{(k,j)}(d)$ (if the variable grows from $V_k$ to $V_j$). Both orientation are allowed and obviously $A_{(k,j)}(d)=-A_{(j,k)}(d)$.
\end{itemize}
In the rest of the section, we will specify the dependence of $\alpha_j(d)$, $\gamma_{jk}$ and
$A_{(j,k)}(d)$ on the matrix $T$ and on the parameter $d$. At first, we introduce the notation for the wave function components on the edges. 
\begin{itemize}
\item For all $j\in\{1,\ldots,n\}$, the wave function on the $j$-th half line is denoted by the symbol $\psi_j$.
\item The wave function on the line connecting the points $V_j$ and $V_k$ is denoted by $\varphi_{(j,k)}$, and obviously one has to take the line orientation into account (we allow both orientations, cf. the paragraph on vector potentials above). We adopt the following convention: If the value $0$ corresponds to the endpoint $V_j$ and the value $\frac{d}{\gamma_{jk}}$ to the endpoint $V_k$, then the wave function is denoted by $\varphi_{(j,k)}$. In the case of the opposite direction, the wave function is denoted by $\varphi_{(k,j)}$. We have $\varphi_{(k,j)}(x)=\varphi_{(j,k)}\left(\frac{d}{\gamma_{jk}}-x\right)$ for all $x\in\left[0,\frac{d}{\gamma_{jk}}\right]$.
\end{itemize}

During the derivation of the expressions for $\alpha_j(d)$, $\gamma_{jk}$ and
$A_{(j,k)}(d)$, we will refrain from indicating the dependence of the parameters $\alpha_j(d),\,
A_{(j,k)}(d)$ on $d$, i.e. we will write simply $\alpha_j,\,A_{(j,k)}$. We will also use symbol $N_j$ standing for the set containing indices of all the edges that are joined to the $j$-th one, i.e. $k\in N_j\Leftrightarrow \gamma_{jk}\neq0$.

The $\delta$ coupling at the endpoint of the $j$-th half line ($j=1, ..., n$) means
\begin{equation}\label{II.serie}
\begin{array}{c}
\psi_j(0)=\varphi_{(j,k)}(0)=\varphi_{(k,j)}\left(\frac{d}{\gamma_{jk}}\right)\quad \text{for all}\; k\in N_j\,,
\\ \vspace{.5em}
\psi_j'(0)+\sum_{k\in N_j}\varphi_{(j,k)}'\left(\frac{d}{\gamma_{jk}}\right)=\alpha_j\cdot\psi_j(0)\,.
\end{array}
\end{equation}
Further relations which will help us to find the parameter
dependence on $d$ come from Taylor expansion. Consider first the
case without any added potential,
\begin{eqnarray}\label{bez pot.}
&& \begin{aligned}
&\varphi^0_{(j,k)}\left(\frac{d}{\gamma_{jk}}\right)
=\varphi^0_{(j,k)}(0)+\frac{d}{\gamma_{jk}}\, \varphi^{0\prime}_{(j,k)}(0)+\O(d^2)\,,\\
&\varphi^{0\prime}_{(j,k)}\left(\frac{d}{\gamma_{jk}}\right)
=\varphi^{0\prime}_{(j,k)}(0)+\O(d)\,.
\end{aligned}
\end{eqnarray}
If the vector potentials are involved, then the relations \eqref{bez pot.} change according to the following lemma.
\begin{lem}\label{potencial}
Let us consider a line parametrized by the variable $x\in[0,L]$,
$L\in(0,+\infty)$, and let $H$ denote a Hamiltonian
of a particle on this line acting as
\begin{equation}\label{H bez v.p.}
H=-\frac{\d^2}{\d x^2}\,.
\end{equation}
We denote by
$\psi^{s,t}$ the solution of $H\psi=k^2\psi$ with the boundary
values $\psi^{s,t}(0)=s$, ${\psi^{s,t}}'(0)=t$. Consider the same
line with a constant vector potential $A$ supported by the interval $[L_1,L_2]$ where $0<L_1<L_2<L$; the Hamiltonian is consequently given by
\begin{equation}\label{H s v.p.}
H_A=\left(-\i\frac{\d}{\d x}-A\right)^2\,.
\end{equation}
Let $\psi_A^{s,t}$ denote the solution of $H_A\psi=k^2\psi$ with
the same boundary values as before, i.e. $\psi_A^{s,t}(0)=s$,
${\psi_A^{s,t}}'(0)=t$. Then the function $\psi_A^{s,t}$ satisfies
$$
\psi_A^{s,t}(L)=\e^{iA(L_2-L_1)}\cdot\psi^{s,t}(L)
\qquad\text{and}\qquad \psi_A^{s,t\prime}(L)
=\e^{iA(L_2-L_1)}\cdot\psi^{s,t\prime}(L)\,.
$$
\end{lem}
The statement can be proven easily, cf. \cite{SMMC99} or \cite{CET10}. Lemma~\ref{potencial} says that the constant vector potential $A$ supported on the line segment of the length $L_2-L_1$ shifts the phase of the wave function and its derivative by $A(L_2-L_1)$. 
Its application on \eqref{bez pot.} gives the relations
\begin{eqnarray}\label{III.serie}
\begin{aligned}
&\varphi_{(j,k)}\left(\frac{d}{\gamma_{jk}}\right) =\e^{\i \frac{d}{2\gamma_{jk}}A_{(j,k)}}\left[\varphi_{(j,k)}(0)+\frac{d}{\gamma_{jk}}\, \varphi_{(j,k)}'(0)\right]+\O(d^2)\,,\\
&\varphi_{(j,k)}'\left(\frac{d}{\gamma_{jk}}\right)=\e^{\i \frac{d}{2\gamma_{jk}}A_{(j,k)}}\cdot\varphi_{(j,k)}'(0)+\O(d)\,.
\end{aligned}
\end{eqnarray}

Equations \eqref{II.serie}
and~\eqref{III.serie} connect values of
the wave functions and their derivatives at all the vertices $V_j$, $j=1, ..., n$. Next we
will eliminate the terms with the ``auxiliary'' functions
$\varphi_{(j,k)}$, and find the relations between $2n$ terms
$\psi_j(0)$, $\psi_j'(0)$, $j=1,\ldots,n$.

We start the elimination by expressing $\varphi_{(j,k)}'(0)$ from the first one of the relations~\eqref{III.serie}, using the fact that $\varphi_{(j,k)}\left(\frac{d}{\gamma_{jk}}\right)=\psi_j(0)$ and $\varphi_{(j,k)}\left(\frac{d}{\gamma_{jk}}\right)=\psi_k(0)$ following from the first equation of~\eqref{II.serie}:
\begin{equation}\label{hvezdicka}
d\: \varphi_{(j,k)}'(0)=\gamma_{jk}\,\e^{-\i \frac{d}{2\gamma_{jk}}A_{(j,k)}}\psi_k(0)
-\gamma_{jk}\psi_j(0)+\O(d^2)\,.
\end{equation}
Then we sum the last equation over $k$ which yields
\begin{equation}\label{suma}
d\sum_{k=1}^n\varphi_{(j,k)}'(0)=\sum_{k=1}^n\gamma_{jk}\,\e^{-\i \frac{d}{2\gamma_{jk}}A_{(j,k)}}\psi_k(0)
-\sum_{k=1}^n\gamma_{jk}\cdot\psi_j(0)+\O(d^2)\,.
\end{equation}
According to the second equation of \eqref{II.serie}, the LHS of Equation~\eqref{suma} is equal to $-d\psi_j'(0)+d\alpha_j\psi_j(0)$. Therefore \eqref{suma} can be rewritten in the form
\begin{equation}\label{j-ty radek}
d\psi_j'(0)=\left(d\alpha_j+\sum_{k=1}^n\gamma_{jk}\right)\psi_j(0)-
\sum_{k=1}^n\gamma_{jk}\,\e^{-\i \frac{d}{2\gamma_{jk}}A_{(j,k)}}\psi_k(0)
+\O(d^2)
\end{equation}
for all $j=1, ..., n$.

Our objective is to choose $\alpha_j(d)$, $\gamma_{jk}$ and
$A_{(j,k)}(d)$ in such a way that in the limit $d\to0$ the system
of relations~\eqref{j-ty radek} with $j=1, ..., n$ tends to the
system of $n$ boundary conditions~\eqref{FTFT}. Recall that the lines
of~\eqref{FTFT} are of two types. We start with the type \eqref{j>=m+1}. corresponding to
$j\geq m+1$. We see immediately that if the parameters $\alpha_j$, $\gamma_{jk}$ and $A_{(j,k)}$ satisfy
\begin{subequations}\label{implicitni}
\begin{gather}
\gamma_{jk}=0\qquad\text{for $k>m$}\,, \\
\gamma_{jk}\,\e^{-\i \frac{d}{2\gamma_{jk}}A_{(j,k)}}=\overline{t_{kj}} \qquad \text{for $k\leq m$}\,, \\
d\alpha_j+\sum_{k=1}^n\gamma_{jk}=1\,,
\end{gather}
\end{subequations}
then \eqref{j-ty radek} leads to
\begin{equation}\label{j>=m+1.}
\psi_j'(0)+\sum_{l=m+1}^n\overline{t_{jl}}\psi'_l(0)
=0 \qquad j=1, ..., m\,,
\end{equation}
in the limit $d\to0$, and therefore Eq.~\eqref{j>=m+1} is satisfied.
Equations~\eqref{implicitni} allow us to express $\alpha_j(d)$, $\gamma_{jk}$ and
$A_{(j,k)}(d)$ for any $j\geq m+1$, $k=1,\ldots,n$:
\begin{itemize}
\item If $k>m$ or $k\leq m\wedge t_{kj}=0$, then $\gamma_{jk}=0$, and thus $V_j$ and $V_k$ are not connected.
\item If $k\leq m$ and $t_{kj}\neq0$, then there is a connecting line between $V_j$ and $V_k$ and it holds
\begin{subequations}\label{spoj}
\begin{gather}
\gamma_{jk}=|t_{kj}|\,, \label{gamma} \\
A_{(j,k)}=2\frac{|t_{kj}|}{d}\arg t_{kj}\,, \label{A} \\
\alpha_j=\frac{1}{d}\left(1-\sum_{k=1}^m |t_{kj}|\right)\,. \label{alpha}
\end{gather}
\end{subequations}
\end{itemize}

Let us proceed to the case $j\leq m$. A
substitution of~\eqref{j-ty radek} with $j\in\{1,\ldots,m\}$ into the left-hand side
of~\eqref{j<=m} leads to the formula
\begin{multline}\label{j-ty radek1}
\psi_j'(0)+\sum_{l=m+1}^n t_{jl}\cdot\psi_l'(0)\\
=\left(\alpha_j+\frac{1}{d}\sum_{h=1}^n\gamma_{jh}\right)\psi_j(0)
-\frac{1}{d}\sum_{k=1}^n\gamma_{jk}\,\e^{-\i \frac{d}{2\gamma_{jk}}A_{(j,k)}}\psi_k(0)+\\
+\sum_{l=m+1}^n
t_{jl}\left[\left(\alpha_l+\frac{1}{d}\sum_{h\in N_l}\gamma_{lh}\right)\psi_l(0)
-\frac{1}{d}\sum_{k=1}^n\gamma_{lk}\,\e^{-\i \frac{d}{2\gamma_{lk}}A_{(l,k)}}\psi_k(0)\right]+\O(d)\,,
\end{multline}
and when we substitute from Equations~\eqref{spoj}, we arrive after a manipulation at
\begin{multline}\label{j-ty radek11}
\psi_j'(0)+\sum_{l=m+1}^n t_{jl}\cdot\psi_l'(0)\\
=\left(\alpha_j+\frac{1}{d}\sum_{h=1}^m\gamma_{jh}+\frac{1}{d}\sum_{h=m+1}^n |t_{jh}|-\frac{1}{d}\sum_{l=m+1}^n t_{jl}\overline{t_{jl}}
\right)\psi_j(0)\\
-\frac{1}{d}\sum_{k=1}^m\left(\gamma_{jk}\,\e^{-\i \frac{d}{2\gamma_{jk}}A_{(j,k)}}
+\sum_{l=m+1}^nt_{jl}\overline{t_{kl}}\right)\psi_k(0)+\O(d)\,.
\end{multline}

Similarly as above, we compare~\eqref{j-ty radek11} with~\eqref{j<=m} and write down the conditions for the missing parameters $\alpha_j$, $\gamma_{jk}$ and $A_{(j,k)}$ ($j\leq m$) that make the right-hand side of~\eqref{j-ty radek11} tend to the right-hand side of~\eqref{j<=m} (i.e. to $0$) as $d\to0$. The conditions are
\begin{subequations}
\begin{gather}
\gamma_{jk}\,\e^{-\i \frac{d}{2\gamma_{jk}}A_{(j,k)}}+\sum_{l=m+1}^nt_{jl}\overline{t_{kl}}=0\,, \\
\alpha_j+\frac{1}{d}\sum_{h=1}^m\gamma_{jh}+\frac{1}{d}\sum_{h=m+1}^n |t_{jh}|-\frac{1}{d}\sum_{l=m+1}^n |t_{jl}|^2=0\,,
\end{gather}
\end{subequations}
hence we derive explicit expressions for $\alpha_j$, $\gamma_{jk}$ and $A_{(j,k)}$ ($j\leq m$):
\begin{itemize}
\item If $\sum_{l=m+1}^nt_{jl}\overline{t_{kl}}=0$, then $\gamma_{jk}=0$, i.e. there is no connecting line between $V_j$ and $V_k$.
\item If $\sum_{l=m+1}^nt_{jl}\overline{t_{kl}}\neq0$, then there is a connection line between $V_j$ and $V_k$, and one can set
\begin{subequations}
\begin{gather}
\gamma_{jk}=\left|\sum_{l=m+1}^n t_{jl}\overline{t_{kl}}\right|\,, \label{gamma_} \\
A_{(j,k)}=\frac{2\i}{d}\left|\sum_{l=m+1}^n t_{jl}\overline{t_{kl}}\right|\cdot\arg\left(-\sum_{l=m+1}^n t_{jl}\overline{t_{kl}}\right)\,, \label{A_} \\
\alpha_j=\frac{1}{d}\left(\sum_{l=m+1}^n |t_{jl}|^2-\sum_{k=1}^m\left|\sum_{l=m+1}^n t_{jl}\overline{t_{kl}}\right|-\sum_{l=m+1}^n |t_{jl}|\right)\,. \label{alpha_}
\end{gather}
\end{subequations}
\end{itemize}

The norm resolvent convergence of this approximation scheme can be proven
in a similar fashion as in \cite{CET10}.
\section{Prospects}

Fulop-Tsutsui family of singular vertex is 
potentially useful to bring about simplicity and solvability
to quantum systems made out of connected graphs.
It has been shown in \cite{HC06, SV09}
that ``quantum chaotic'' features are found in a one dimensional
system with a Fulop-Tsutsui point interaction.
It might be possible to show the existence of spectral fluctuation 
usually associated to chaotic system in two and three dimension
in a star graph with Fulop-Tsutsui vertex of degree as small as $n=3$. 
%

\section*{Acknowledgments}

The research was supported  by the 
Japanese Ministry of Education, Culture, Sports, Science and Technology 
under the Grant number 21540402.


\end{document}